\documentclass[5p,times]{elsarticle}
\journal{European Journal of Control}

\usepackage{textcomp}
\usepackage{color}
\usepackage{generic}
\usepackage{graphicx}
\usepackage{subcaption}
\usepackage{etoolbox} 
\usepackage{scalerel} 
\usepackage{siunitx}
\usepackage{pifont}
\usepackage{graphicx}
\usepackage{amsmath,amssymb,textcomp}
\usepackage{algorithm} 
\usepackage{cite}
\usepackage{booktabs}
\usepackage{url}
\usepackage{algpseudocode}
\usepackage{paralist}
\usepackage{color}
\usepackage[T1]{fontenc}

\usepackage{stmaryrd}
\usepackage{epsfig} 
\usepackage{bm} 
\usepackage{listings,lstautogobble}
\usepackage{mathtools}
\usepackage{mathrsfs}
\usepackage{verbatim} 
\usepackage{epstopdf}
\usepackage{listings}
\usepackage{parcolumns}
\usepackage{color}
\usepackage{mathtools}
\usepackage{mathrsfs}
\usepackage{xspace}
\usepackage{verbatim}
\usepackage[utf8]{inputenc}
\usepackage{calrsfs}
\usepackage{rotating,multirow}
\usepackage{mathtools}

\usepackage[most]{tcolorbox}

\usepackage[T1]{fontenc}
\usepackage[utf8]{inputenc}
\usepackage{hyperref}
\usepackage{cleveref}

\usepackage{amsthm}
\newtheorem{definition}{Definition}

\newtheorem{theorem}{Theorem}
\newtheorem{remark}{Remark}
\newtheorem{lemma}{Lemma} 
\newtheorem{assump}{Assumption}

  \renewcommand{\footnotesize}{\fontsize{8pt}{11pt}\selectfont}

\DeclareMathAlphabet{\mathcal}{OMS}{cmsy}{m}{n}

\DeclarePairedDelimiter{\norm}{\lVert}{\rVert}

\DeclareMathOperator{\rank}{rank}

\usepackage{hyperref}

\newcommand{\x}{\textbf{x}}

\def \Int{\mathrm{int}}

\allowdisplaybreaks

\begin{document}

\begin{frontmatter}

\title{Robust Data-Driven Predictive Control of Unknown 
Nonlinear Systems\\ using Reachability Analysis}

\author[add1]{Mahsa Farjadnia\corref{cor1}}
\ead{mahsafa@kth.se}
\author[add2]{Amr Alanwar}
\ead{aalanwar@constructor.university}
\author[add3,add4]{Muhammad Umar B. Niazi}
\ead{niazi@mit.edu}
\author[add1]{Marco Molinari}
\ead{marcomo@kth.se}
\author[add3]{Karl Henrik Johansson}
\ead{kallej@kth.se}
\address[add1]{Department of Energy Technology, School of Industrial Engineering and Management, KTH Royal Institute of Technology, Stockholm, Sweden}
\address[add2]{ School of Computer Science and Engineering, Constructor University, Bremen, Germany}
\address[add3]{Division of Decision and Control Systems, School of Electrical Engineering and Computer Science, KTH Royal Institute of Technology, Stockholm, Sweden}
\address[add4]{Laboratory for Information and Decision Systems, Massachusetts Institute of Technology, Cambridge, USA}

\cortext[cor1]{Corresponding author.}

\begin{abstract}
This work proposes a robust data-driven predictive control approach for unknown nonlinear systems in the presence of bounded process and measurement noise. Data-driven reachable sets are employed for the controller design instead of using an explicit nonlinear system model. Although the process and measurement noise are bounded, the statistical properties of the noise are not required to be known. By using the past noisy input-output data in the learning phase, we propose a novel method to over-approximate exact reachable sets of an unknown nonlinear system. Then, we propose a data-driven predictive control approach to compute safe and robust control policies from noisy online data. The constraints are guaranteed in the control phase with robust safety margins by effectively using the predicted output reachable set obtained in the learning phase. Finally, a numerical example validates the efficacy of the proposed approach and demonstrates comparable performance with a model-based predictive control approach.
\end{abstract}

\end{frontmatter}

\section{Introduction}
Model Predictive Control (MPC) is one of the most powerful and well-established methods in the field of control systems due to its ability to efficiently handle system constraints, nonlinear dynamics, and trajectory tracking \citep{rawlings2012postface}. The MPC scheme generally requires a sufficiently accurate dynamical system mo-del to guarantee optimal control performance while satisfying system constraints. However, acquiring an accurate model in practice, especially for complex systems, is often time-con-suming and costly \citep{8745685}. 
Consequently, data-driven approaches for predictive control have gained significant interest in the control community in recent years \citep{hou2013model}. 

A remarkable contribution to data-driven predictive control for linear time-invariant (LTI) systems can be found in \citep{coulson2019data}, where the authors addressed an optimal trajectory tracking problem for unknown LTI systems using the data-enabled predictive control (DeePC) algorithm. DeePC relies on behavioral systems theory and Willems’ Fundamental Lemma (WFL) \citep{willems2005note} to learn a non-parametric system model from noise-free data. Although DeePC is a seminal approach, it was initially developed only for noise-free systems, and later it was extended to stochastic LTI systems in \citep{coulson2021distributionally}. However, when the statistics of the uncertainties and noise in the LTI systems are unknown but bounded, the data-driven set-based approach presented in \citep{alanwar2021robust} is more appropriate for guaranteeing robust constraint satisfaction. This approach relies on using data-driven reachable set prediction within a predictive control scheme to compensate for the complexity of capturing an accurate model. Unlike LTI systems, data-driven control for nonlinear systems is quite challenging, and there are many open challenges in this research area.

Recent works on data-driven predictive control for uncertain nonlinear systems include \citep{hewing2019cautious}, which provides a chance-constrained MPC approach to control a nonlinear system with a known nominal part and unknown additive dynamics modeled as a Gaussian process. However, the nonlinear system is not completely unknown in \citep{hewing2019cautious}; in addition, the proposed approach is not robust and allows a certain amount of constraint violations. Extensions of the WFL to nonlinear systems have also been proposed to make it analogous to data-driven control approaches in LTI systems. In this regard, \citep{berberich2022linear} proposes a data-driven predictive control approach for unknown nonlinear and control-affine systems using available noise-free data. This approach relies on local linearization of the underlying system using noise-free data in order to apply WFL. In \citep{lian2021nonlinear}, WFL is used to reproduce kernel Hilbert space and extend its applicability to a class of nonlinear systems. However, it is important to emphasize that all the works cited above either assume noise-free data or do not provide robust safety guarantees under system constraints.

In this paper, we provide an extension of the results
on robust zonotopic data-driven predictive control presented in \citep{alanwar2021robust} to
unknown nonlinear systems under bounded process and measurement noise. 
Our proposed approach uses noisy data to derive an implicit data-driven model of the system in real-time and obtains an optimal control input guaranteeing robust constraint satisfaction. To this end, we employ a zonotopic set representation to provide an implicit data-driven system model.

Our proposed approach comprises learning and control pha-ses. The goal of the learning phase is to compute an implicit data-driven representation of the unknown nonlinear system at each time step using zonotopes. In this phase, a data-driven linear model approximates the unknown nonlinear system by employing Taylor series expansion and utilizing the past input-output data over a finite horizon. Subsequently, the model mismatch and the Lagrange remainder term of the Taylor series are bounded by a zonotope using the available data. In the control phase,  we propose a robust data-driven approach called \textit{nonlinear zonotopic predictive control} (NZPC). NZPC utilizes the learning phase and zonotope recursion at each time step to predict the reachable output set over a finite horizon. The optimal control problem solved by NZPC in this phase yields an optimal control input that minimizes the given cost function and satisfies the specified constraints. This method involves updating the input-output data set as the closed-loop system evolves over time, i.e., old data is discarded when new data is collected and added to the data set, thereby enhancing the implicit system representation and improving the controller's performance. The code to recreate our findings is publicly available\footnotemark.

\footnotetext{\href{https://github.com/aalanwar/Data-Driven-Predictive-Control}{github.com/aalanwar/Data-Driven-Predictive-Control}}

This paper presents three main contributions that address the data-driven predictive control problem for unknown nonlinear systems. Firstly, we propose Algorithm~\ref{alg: Reachability}, which efficiently estimates the reachable sets of an unknown nonlinear system using zonotopes, leveraging noisy input-output data.
Secondly, we introduce Theorem~\ref{theorem: Reach}, which provides a formal proof that the reachable sets obtained through Algorithm~\ref{alg: Reachability} offer an over-approximation of the exact reachable sets for nonlinear systems. This theoretical result establishes the reliability, accuracy, and practical applicability of Algorithm~\ref{alg: Reachability}. Finally, we present Algorithm~\ref{alg: zonopc}, which computes safe and robust control inputs for a receding horizon optimization problem associated with an unknown nonlinear system. This algorithm provides a principled approach to handling uncertain and dynamic environments, ensuring robust constraint satisfaction. Theorem~\ref{theorem: cont} provides the formal proof that the control inputs generated by Algorithm~\ref{alg: zonopc} robustly satisfy the system's constraints, thus providing a safe and reliable control strategy.

The remainder of this paper is organized as follows. Section~\ref{sec: pb} provides preliminaries and formally defines the problem. In Section~\ref{sec: ZPC}, we present our robust data-driven predictive control approach for unknown nonlinear systems. An illustrative example demonstrating the proposed approach is presented in Section~\ref{sec: Example}, and its performance is compared to a robust model predictive control approach in the same section. We conclude the paper with Section~\ref{sec: con}, where we summarize our results and discuss the future prospects of this work.

\section{Preliminaries and Problem Statement} \label{sec: pb}
\subsection{Notations}
The set of natural numbers, non-negative integers, and real $n$-dimensional space are denoted by $\mathbb{N}$, $\mathbb{Z}_{\geq 0}$, and $\mathbb{R}^{n}$, respectively. 
The absolute value of a scalar $c$ is denoted as $\lvert c \rvert$. For an $n \times 1$ vector $x \in \mathbb{R}^{n}$, $\Vert x \Vert_\infty$ and $\Vert x \Vert_2$ respectively represent the infinity-norm and the 2-norm of $x$, and $\Vert x \Vert_P^2=x^{\top}P x$. Furthermore, $\mathbb{R}^{n \times m}$ denotes the real $n \times m$ matrix space. We write $1_n$, and $0_n$ for n-dimensional column vectors with all entries equal to 1 and 0, respectively. For simplicity, where no confusion arises, we use 0 and 1 to denote matrices with the proper dimensions. 
For a given vector $x$, $x^{(i)}$ denotes its $i$-th element. For a given matrix $A$, the row $i$ of $A$ is denoted by $(A)_{i,.}$, the column $j$ of $A$ is denoted by $(A)_{.,j}$, and the element at row $i$ and column $j$ of $A$ is presented by $(A)_{i,j}$. Finally,  $\mathrm{diag}(\cdot)$ denotes the diagonal operator, which constructs a diagonal matrix using its arguments. For $X \in \mathbb{R}^{n \times m}$, $X^{\top}$ and $X^\dagger$ indicate the transpose and Moore-Penrose pseudoinverse of a matrix $X$, respectively. If matrix $X$ is full-row rank, then $X^\dagger = X^{\top} (XX^{\top})^{-1}$ is the right inverse, i.e., $XX^\dagger = I_n$, where $I_n \in \mathbb{R}^{n\times n}$ is the identity matrix. The Kronecker product is denoted by $\otimes$. A stacked window of sequence $\{x(k)\}_{k=t-T}^{t}$ is
\begin{equation*}
    X_{[t-T,t]}= \begin{bmatrix}  x(t-T)& x(t-T+1) & \cdots & x(t)\end{bmatrix}.
\end{equation*}
For notational simplicity, we introduce $X_+=X_{[t-T+1,t]}$ and $X_-=X_{[t-T,t-1]}$ to represent the shifted stacked windows. We use the subscript ${t+k|t}$ to highlight predictive quantities, e.g., $x_{t+k|t}$ is the $(t+k)$-step-ahead prediction of the vector $x$ initialized at time $t$, and $x_{t|t}=x(t)$. 

\subsection{Set Representations}
A zonotope is an affine transformation of a unit hypercube, \citep{kuhn1998rigorously}, and it is defined as follows.
\begin{definition} [\textbf{Zonotope}]
    Given a center $c_{\mathcal{Z}} \in \mathbb{R}^{n} $ and a number $\gamma_\mathcal{Z}\in \mathbb{N}$ of generator vectors $g^{i}_{\mathcal{Z}}\in\mathbb{R}^n$, a \textit{zonotope} is defined as
    \begin{equation*}
        \mathcal{Z} = \Big\{ x \in \mathbb{R}^{n} \Big| x = c_{\mathcal{Z}} + \sum_{i=1}^{\gamma_\mathcal{Z}} \beta^{(i)} g^{i}_{\mathcal{Z}},-1 \leq \beta^{(i)}\leq 1 \Big\}.
    \end{equation*}
    To simplify notation, we denote a zonotope as $ \mathcal{Z} = \langle c_{\mathcal{Z}},G_{\mathcal{Z}} \rangle$, where $G_{\mathcal{Z}} = [g^{1}_{\mathcal{Z}}\dots g^{\gamma_\mathcal{Z}}_{\mathcal{Z}}]\in \mathbb{R}^{n\times \gamma_\mathcal{Z} }$. \hfill $\lrcorner$
\end{definition}

A zonotope $\mathcal{Z}$ could be over-approximated by a multidimensional interval as \citep[Proposition~ 2.2]{althoff2010reachability}: 
\begin{equation*}
    I = \Int(\mathcal{Z}) = [c_{\mathcal{Z}}-\Delta g,c_{\mathcal{Z}}+\Delta g], \quad \Delta g=\sum_{i=1}^{\gamma_\mathcal{Z}} |g^{i}_{\mathcal{Z}}|
\end{equation*}
where the absolute value is taken element-wise. The conversion of an interval $I=[\underline{I},\overline{I}]$ to a zonotope is denoted by $\mathcal{Z} =  \mathrm{zonotope}(\underline{I},\overline{I})$.

Given $L\in\mathbb{R}^{m\times n}$, the linear transformation of a zonotope $\mathcal{Z}$ is a zonotope $L\mathcal{Z}=\langle Lc_{\mathcal{Z}},LG_{\mathcal{Z}} \rangle$. Given two zonotopes $\mathcal{Z}_1 = \langle c_{\mathcal{Z}_1 },G_{\mathcal{Z}_1 } \rangle$ and $\mathcal{Z}_2 = \langle c_{\mathcal{Z}_2 },G_{\mathcal{Z}_2 } \rangle$, the Minkowski sum is defined and obtained as
\begin{align*}
     \mathcal{Z}_1 \oplus \mathcal{Z}_2  
    &= \{ z_1 + z_2 \,|\, z_1 \in \mathcal{Z}_1, z_2 \in \mathcal{Z}_2 \} \nonumber\\
    &= \Big\langle c_{\mathcal{Z}_1 }+ c_{\mathcal{Z}_1 },[G_{\mathcal{Z}_1 } \, \, G_{\mathcal{Z}_2 }] \Big\rangle.
\end{align*}
For simplicity, we use the notation $+$ instead of $\oplus$ to represent the Minkowski sum and write $\mathcal{Z}_1 - \mathcal{Z}_2$ to denote $\mathcal{Z}_1 + (-\mathcal{Z}_2)$. The Cartesian product of two zonotopes $\mathcal{Z}_1$ and $\mathcal{Z}_2$ is defined and obtained as
\begin{align*}
    \mathcal{Z}_1 \times \mathcal{Z}_2 
    &= \left\{ {\left. {\left[ {\begin{array}{*{20}{c}}
    {{z_1}}\\
    {{z_2}}
    \end{array}} \right]} \right|{z_1} \in {\mathcal{Z}_1},{z_2} \in {\mathcal{Z}_2}} \right\}\nonumber \\ 
    &= \left\langle {\left[ {\begin{array}{*{20}{c}}
    {{c_{{\mathcal{Z}_1}}}}\\ 
    {{c_{{\mathcal{Z}_2}}}}
    \end{array}} \right],\left[ {\begin{array}{*{20}{c}}
    {{G_{{\mathcal{Z}_1}}}}&0\\
    0&{{G_{{\mathcal{Z}_2}}}}
    \end{array}} \right]} \right\rangle.
\end{align*}

A matrix zonotope is defined as follows \citep[p.~52]{althoff2010reachability}.

\begin{definition}[Matrix Zonotope]
    Given a center matrix $C_{\mathcal{M}} {\in} \mathbb{R}^{n\times T} $ and a number $\gamma_{\mathcal{M}} {\in}\mathbb{N}$ of generator matrices $G^{1}_{\mathcal{M}},\dots,G^{\gamma_{\mathcal{M}}}_{\mathcal{M}}{\in}\mathbb{R}^{n\times T }$, a \textit{matrix zonotope} is defined as 
    \begin{equation*}
        \mathcal{M} = \Big\{ X \in \mathbb{R}^{n\times T}\Big| X = C_{\mathcal{M}} + \sum_{i=1}^{\gamma_\mathcal{M}} \beta^{(i)} G^{i}_{\mathcal{M}},-1 \leq \beta^{(i)}\leq 1 \Big\}.
    \end{equation*}
    We use the short notation $ \mathcal{M} = \langle C_{\mathcal{M}},G_{\mathcal{M}} \rangle$ to denote a matrix zonotope, where $G_{\mathcal{M}} = \{ G^{1}_{\mathcal{M}},\dots,G^{\gamma_{\mathcal{M}}}_{\mathcal{M}} \} $. \hfill $\lrcorner$
\end{definition}
\subsection{Problem Statement}
Consider a discrete-time nonlinear control system 
\begin{subequations}\label{eq: sys}
\begin{align} 
x(k + 1) &= f(x(k),u(k)) + w(k) \label{eq: state eq}\\
y(k) &= Hx(k) + v(k)  \label{eq: output eq }
\end{align}
\end{subequations}
where $f: \mathbb{R}^{n_x} \times \mathbb{R}^{n_u} \xrightarrow{}\mathbb{R}^{n_x}$ is an unknown nonlinear function, $x(k) \in \mathcal{X}\subset\mathbb{R}^{n_x}$ and $y(k) \in \mathbb{R}^{n_y}$ are respectively the state and the output of the system with $n_y \leq n_x$ at time $k \in \mathbb{Z}_{\geq 0}$,
 $u(k) \in \mathcal{Z}_u = \langle c_{\mathcal{Z}_u},G_{\mathcal{Z}_u} \rangle \subset \mathbb{R}^{n_u}$ is the control input, $\mathcal{X}_0 \subset  \mathcal{X}$ is the set of initial states, $H \in \mathbb{R}^{n_y\times n_x}$ is the system output matrix which is assumed to be known and full row-rank, i.e., $\rank(H)=n_y$, $w(k)$ denotes the process noise bounded by a zonotope $w(k)\in \mathcal{Z}_w = \langle c_{\mathcal{Z}_w},G_{\mathcal{Z}_w} \rangle$, and $v(k)$ is the measurement noise bounded by a zonotope $v(k)\in \mathcal{Z}_v = \langle c_{\mathcal{Z}_v},G_{\mathcal{Z}_v} \rangle$ for all time steps. At each time step $k$, input and output constraints for a controller design are given by 
\begin{align}\label{eq: constraint}
    u(k) \in \mathcal{U}_k\subseteq \mathcal{Z}_u , \quad y(k) \in \mathcal{Y}_k \subset \mathbb{R}^{n_y}.
 \end{align}

In \eqref{eq: constraint}, the set $\mathcal{Z}_u $ represents the time-invariant domain of control inputs, corresponding to constraints inherent to the problem and may be driven by physical limitations. Meanwhile, the sets $\mathcal{U}_k$ and $\mathcal{Y}_k$ are time-varying input and output constraints, respectively, because they account for additional limitations dependent on time-varying circumstances.
  
We aim to solve a receding horizon optimal control problem by employing only the input-output data of the system \eqref{eq: sys} without explicit knowledge of the nonlinear function $f(x(k),u(k))$. To be precise, we assume that at each time step $t \in \mathbb{Z}_{\geq 0}$, we have access to the past $K$ input-output trajectories of different lengths $T_i+1$, $i=1,\dots,K$, denoted by $\{ u(k)\}^{t}_{k=t-T_i}$ and $\{ y(k)\}^{t}_{k=t-T_i}$. For simplicity, we consider a single trajectory and collect all the input and `noisy' output data in the following matrices
\begin{align*}
     Y_{[t-T,t]} &= \begin{bmatrix}  y(t-T)& y(t-T+1) & \cdots & y(t)\end{bmatrix}, \\
      U_{[t-T,t]} &= \begin{bmatrix} u(t-T) &  u(t-T+1)  & \cdots & u(t) \end{bmatrix}.
 \end{align*}
If a data point index $k$ is negative, it refers to the data obtained from offline experiments, i.e., it was collected before implementing NZPC in the control loop. We indicate the set of all available data at time $t$ by $D_{[t-T,t]}=( U_{[t-T,t]},Y_{[t-T,t]})$.

Given the set of all possible inputs $u(k) \in \mathcal{Z}_{u}$, process noise $w(k)\in \mathcal{Z}_w$, for $ k = 0,\dots,N-1$, and initial state set $\mathcal{X}_0$, the state reachable set $\mathcal{R}^x_N$ is defined as the set of all possible states in the state space that a system can reach after $N$ time steps:
     \begin{align} \label{eq:state_reachable_set}
        \mathcal{R}^x_N = \big\{ &x(N) \in \mathbb{R}^{n_x} \,|\, x(k+1) = f(x(k),u(k)) + w(k),\nonumber \\ 
        &\forall x(0)\in \mathcal{X}_0,w(k)\in \mathcal{Z}_w, \text{ and } u(k) \in \mathcal{Z}_{u},\nonumber\\
        &k =0,\dots,N-1 \big\}.
    \end{align}
In this paper, however, we are particularly interested in the output reachable set, which is defined below. 
\begin{definition}[\textbf{Output reachable Set}]
    The output reachable set $\mathcal{R}^y_N$ is defined as the set of all possible outputs that a system can reach after $N$ time steps:
     \begin{equation*}
        \mathcal{R}^y_N = H\mathcal{R}^x_N+\mathcal{Z}_v
    \end{equation*}
    where $\mathcal{R}^x_N$ is the state reachable set defined in \eqref{eq:state_reachable_set} and $\mathcal{Z}_v$ is the measurement noise zonotope. 
    \hfill $\lrcorner$
\end{definition}
We denote the sequences of the unknown actual process and measurement noise corresponding to the available input-output trajectories as $\{w(k)\}_{k=t-T}^{t}$ and $\{v(k)\}_{k=t-T}^{t}$, respectively. Based on the assumptions of bounded noise, we can directly infer that the stacked matrices $W_{[t-T,t-1]}$, $V_{[t-T,t-1]}$, and $V_{[t-T+1,t]}$ are bounded, satisfying the conditions $W_{[t-T,t-1]} \in\mathcal{M}_w$ and
$V_{[t-T,t-1]},V_{[t-T+1,t]} \in \mathcal{M}_v$. Here, $\mathcal{M}_w = \langle C_{\mathcal{M}_w},G_{\mathcal{M}_w} \rangle$ and $\mathcal{M}_v=\langle C_{\mathcal{M}_v},G_{\mathcal{M}_v} \rangle$ represent matrix zonotopes resulting from the concatenation of multiple noise zonotopes $\mathcal{Z}_w$ and $\mathcal{Z}_v$ as described in \citep{alanwar2023data}. Additionally, we denote the actual state trajectory, which is unknown, as $X_{[t-T,t]}$.
\subsection{Main Assumptions}
In this subsection, we present the essential assumptions for the data-driven nonlinear zonotopic predictive control approach.

\begin{assump}\label{asmp: normstate}
   We assume that the zonotope $\mathcal{Z}_\eta = \langle 0_{n_x}, \eta I_{n_x} \rangle$ inscribes the state space $\mathcal{X}$, where $\eta>0$ is known. In other words, for any $k \in \mathbb{Z}_{\geq 0}$, we assume $\Vert x(k) \Vert_\infty \le \eta$. 
   \hfill $\lrcorner$
\end{assump}

It is important to note that the aforementioned assumption does not imply that the system is already stabilized. Rather, it is a consequence of the operational or physical constraints on the system states; that is, the domain of states is a bounded subset $\mathcal{X}$. 

We define $\mathcal{M}_\eta=\langle C_{\mathcal{M}_\eta},G_{\mathcal{M}_\eta} \rangle$ to denote a matrix zonotope resulting from the concatenation of multiple zonotopes $\mathcal{Z}_\eta$. Additionally, we define 
 \[
 \mathcal{F} = \mathcal{Z}_\eta \times \mathcal{Z}_{u}.
 \]
 For simplicity, we omit the time index $k$ when no confusion may arise. In order to introduce a concise notation, we define a new extended state vector as follows: 
 \begin{equation*}
     \xi = \begin{bmatrix} x \\ u\end{bmatrix}.
 \end{equation*}
\begin{assump}\label{asmp: differentiable}
 The unknown nonlinear function $f$ is twice continuously differentiable in $x$ and $u$.
 \hfill $\lrcorner$
\end{assump}
 According to Assumption~\ref{asmp: differentiable}, $f$ is locally Lipschitz continuous in $x$ and $u$ at each dimension, i.e., there exists a vector of constants $L_f \in \mathbb{R}^{n_x}$ with $L_f^{(i)} > 0$, for $i = 1,\dots,n_x$, such that for any $\xi_1,\xi_2 \in \mathcal{F}$, it holds that 
 \[\vert f^{(i)}(\xi_1)-f^{(i)}(\xi_2)\vert \le L_f^{(i)} \Vert \xi_1-\xi_2\Vert_2.\]
 Considering a separate Lipschitz constant for each element of the function $f$ decreases conservatism, especially when the data has a different scale for each dimension.

\section{Robust Data-driven Predictive Control}\label{sec: ZPC}
In this section, we propose a data-driven nonlinear zonotopic predictive control (NZPC) approach for the nonlinear system~\eqref{eq: sys}. The proposed NZPC algorithm consists of two online phases: the learning phase and the control phase, which are described in detail in the following subsections. 
\subsection{Learning Phase} \label{subsec: Learning}
In this phase, we propose a data-driven algorithm to over-approximate reachable sets of a nonlinear control system given input and noisy output data $D_{[t-T,t]}$. This algorithm is a key technical result for predicting all possible output reachable trajectories within the prediction horizon in the NZPC approach. Using \eqref{eq: state eq}, we rewrite the output equation~\eqref{eq: output eq } as 
\begin{equation}\label{eq: in-out}
    y(k+1) = f_H(x(k),u(k)) + Hw(k) + v(k+1)
\end{equation}
with 
\[
 f_H(x(k),u(k)) = Hf(x(k),u(k)).
 \]
Our NZPC approach does not involve state measurement $x(k)$, and, therefore, we need to construct a zonotope that contains all possible $x(k)$ consistent with the output measurement $y(k)$.
\begin{lemma}[\citep{alanwar2021dataSIR}]
\label{Lemma: normstate}
    Let Assumption~\ref{asmp: normstate} hold. Then, given the output measurement $y(k)$ of the system~\eqref{eq: sys}, it holds that $x(k) \in \mathcal{Z}_{x|y}$, where 
    \begin{align*}
         \mathcal{Z}_{x|y} &= H^\dagger(y(k)-\mathcal{Z}_{v}) + g_\eta \mathcal{Z}_{\eta}\\
         &=\langle c_{x|y}, G_{x|y}\rangle
    \end{align*}
    with $g_\eta = (I_{n_x}-H^\dagger H)$, $c_{x|y} = H^\dagger(y(k)-c_{\mathcal{Z}_v})$, and $G_{x|y} = \begin{bmatrix} H^\dagger G_{\mathcal{Z}_v} & \eta g_\eta \end{bmatrix}$.  
\end{lemma}

Let $\xi^\star = \begin{bmatrix} {x^{\star}}^\top & {u^\star}^\top \end{bmatrix}^\top$. Then, Taylor series of $f_H$ around $\xi^\star$ can be represented (see \citep{berz1998computation}) by  
\begin{align}
f_H(\xi) &= f_H(\xi^\star) + \frac{\partial f_H(\xi)}{\partial \xi}\Big|_{\xi=\xi^\star} (\xi- \xi^\star)+L_H
\label{eq: Taylor_f_or}
\end{align}
where $L_H$ is the Lagrange remainder given by
\begin{equation*}
    L_H^{(i)} = \frac{1}{2}(\xi - \xi^\star)^T \frac{\partial^2 f_H^{(i)}(z)}{\partial \xi^2}\Big|_{\xi=\xi^\star}(\xi - \xi^\star)
\end{equation*}
for some $z\in\{ \xi^\star+\alpha (\xi-\xi^\star)| \alpha \in [0,1]\}$.

To present a standard notation of the linearized system, we can separate vector $\xi$ into $x$ and $u$ in \eqref{eq: Taylor_f_or} and rewrite it as follows:
\begin{equation}
f_H(x,u)= \underbrace{\begin{bmatrix}f_H(x^\star,u^\star) & A_H & B_H\end{bmatrix}}_{M_H} \begin{bmatrix}1\\x-x^\star\\u-u^\star\end{bmatrix} + L_H
\label{eq: Taylor_f}
\end{equation}
where
\begin{align*} 
   A_H  = \frac{\partial f_H(x,u)}{\partial x}\Big|_{\begin{smallmatrix}x=x^\star\\u=u^\star\end{smallmatrix}}, \,\, B_H = \frac{\partial f_H(x,u)}{\partial u}\Big|_{\begin{smallmatrix}x=x^\star\\u=u^\star\end{smallmatrix}}.
\end{align*}

If a model of a nonlinear system is available, the Lagrange remainder $L_H$  could be over-approximated as described in \citep[Section 3.4.3]{althoff2010reachability}. However, in this paper, the model is unknown. Therefore, we over-approximate $L_H$ from the available input-output data using Algorithm~\ref{alg: Reachability}, which is described below. 

\begin{algorithm}[!]
  \caption{Data-Driven Reachability Analysis}
  \label{alg: Reachability}
  \textbf{Input}: Input-output data $D_{[t-T,t]}$, initial set $\mathcal{R}^y_{t|t}$, process noise zonotope $\mathcal{Z}_w $, measurement noise zonotope $\mathcal{Z}_v$, Lipschitz constants $L_f^{(i)}$ for $i = 1,\dots, n_x$, covering radius $\delta$, upper bound $\eta$, and input zonotope $\mathcal{Z}_u$.\\
  \textbf{Output}: Data-driven reachable sets $ \hat{\mathcal{R}}^y_{t+k+1|t}$ for $k{=}0,\dots,N-1$.
  \begin{algorithmic}[1]
    \State $\hat{M}_H = {\small (Y_{+}- H C_{\mathcal{M}_w}-C_{\mathcal{M}_v}) \begin{bmatrix} 1^{\top}_{T} \\ H^\dagger (Y_{-} - 1^{\top}_{T} \otimes  y^\star)  \\ U_{-} -  1^{\top}_{T}\otimes u^\star  \end{bmatrix}^\dagger}$ \label{eq: LinearMatrices}
    \State $ \underline{z}_l = \displaystyle \min_j \Bigg( {(Y_{+})}_{.,j} - \hat{M}_H \begin{bmatrix}0\\ 
    H^\dagger ({(Y_{-})}_{.,j}- (y^\star - v^\star))\\ {(U_{-})}_{.,j} - u^\star\end{bmatrix} \Bigg)$\label{eq: ZL_low}
    \State  $
      \overline{z}_l = \displaystyle \max_j \Bigg( {(Y_{+})}_{.,j} - \hat{M}_H \begin{bmatrix}0\\ 
    H^\dagger ({(Y_{-})}_{.,j}- (y^\star - v^\star))\\ {(U_{-})}_{.,j} - u^\star\end{bmatrix} \Bigg)$\label{eq: ZL_upp}
   \State $\mathcal{Z}_L = \mathrm{zonotope}(\underline{z}_l,\overline{z}_l){-}  \hat{M}_H \begin{bmatrix}1\\ 
    {-}H^\dagger \mathcal{Z}_v{+}g_\eta \mathcal{Z}_\eta\\ 0_{n_u}\end{bmatrix} -H^\dagger \mathcal{Z}_w {-} \mathcal{Z}_v$\label{eq: ZL}
     \State $ \mathcal{Z}_\epsilon = \Big\langle 0_{n_y},\mathrm{diag}(\vert (H)_{1,.}\vert L_f \delta/2,\dots,\vert (H)_{n_y,.}\vert L_f\delta/2)\Big\rangle$ \label{eq: algZeps}
    \For{$k = 0:N-1$}\label{eq: foralg1}
  \State\!\!\!\!\!$\hat{\mathcal{R}}^{x}_{t+k|t} = H^\dagger (\hat{\mathcal{R}}^y_{t+k|t} - \mathcal{Z}_v) + g_\eta \mathcal{Z}_\eta$ \label{eq :algRx}
  \State\!\!\!\!\!$\hat{\mathcal{R}}^y_{t+k+1|t} = \hat{M}_H (1 {\times}  ((\hat{\mathcal{R}}^{x}_{t+k|t} {\times} \mathcal{Z}_u) {-}\xi^\star)) {+} \mathcal{Z}_v {+} H \mathcal{Z}_w {+}\mathcal{Z}_L {+} \mathcal{Z}_\epsilon$ \label{eq: Y_Reach_D}
  \EndFor \label{eq: endforalg1}
  \end{algorithmic}
\end{algorithm}

Algorithm~\ref{alg: Reachability} computes over-approximated data-driven reachable sets of the nonlinear control system \eqref{eq: sys}. At time $t$, the over-approximated data-driven reachable set and the exact reachable set within $k$ steps are denoted by $\hat{\mathcal{R}}^y_{t+k|t}$ and $\mathcal{R}^y_{t+k|t}$, respectively. The initial output set is denoted by $\mathcal{R}^y_{t|t}$. The procedure of Algorithm~\ref{alg: Reachability} is as follows: Set $k=0$. First, we obtain an approximate linearized model $\hat{M}_H$ given available data $D_{[t-T,t]}$ in line \ref{eq: LinearMatrices}; \citep{alanwarl4dc,alanwar2023data}. Subsequently, for the chosen linear model, we derive a zonotope $\mathcal{Z}_L$ that over-approximates the modeling mismatch $\Delta M_H$ together with the Lagrange remainder $L_H$ following lines \ref{eq: ZL_low}-\ref{eq: ZL}. By computing a zonotope $\mathcal{Z}_\epsilon$, we over-approximate the modeling mismatch along with the Lagrange remainder for all data points in $\mathcal{F}$ in line \ref{eq: algZeps}. Eventually, we compute the over-approximated output reachable sets $\hat{\mathcal{R}}^y_{t+k+1|t}$ for $k = 0,\dots, N-1$ according to lines \ref{eq: foralg1}-\ref{eq: endforalg1}. 

 \begin{theorem}\label{theorem: Reach}
 Let Assumptions~\ref{asmp: normstate} and \ref{asmp: differentiable} hold. Then, given the input-output trajectories $D_{[t-T,t]}$ of the system \eqref{eq: sys} at time $t\in \mathbb{Z}_{\geq 0} $, the data-driven reachable set $\hat{\mathcal{R}}^y_{t+k+1|t}$ computed in Algorithm~\ref{alg: Reachability} over-approximates the exact reachable set $\mathcal{R}^y_{t+k+1|t}$, i.e., $\hat{\mathcal{R}}^y_{t+k+1|t}  \supseteq \mathcal{R}^y_{t+k+1|t}$, holds for for $k = 0,\dots, N-1$. 
 \end{theorem}

 \begin{proof}
 We can rewrite \eqref{eq: Taylor_f} as follows:
 \begin{equation}\label{eq: app_f}
   f_H(\xi) = \hat{M}_H\begin{bmatrix}1\\\xi-\xi^\star\end{bmatrix}+\Delta M_H\begin{bmatrix}1\\\xi-\xi^\star\end{bmatrix} + L_H  
 \end{equation}
where $\hat{M}_H$ is an approximation of $M_H$ and $\Delta M_H = M_H -\hat{M}_H$ is the model mismatch. We aim to obtain $\hat{M}_H$, which is an approximate linearization of $f_H(x(k),u(k))$, and over-approximate the model mismatch $\Delta M_H$ along with the Lagrange remainder $L_H$. Let Assumption~\ref{asmp: differentiable} hold. Then, by neglecting the Lagrange remainder and substituting \eqref{eq: Taylor_f} in \eqref{eq: in-out}, the following holds for the collected data and its corresponding unknown actual state and noise realizations:
  \begin{equation}\label{eq: Y_data}
      Y_{+} \approx M_H 
       \begin{bmatrix} 1^{\top}_{T} \\  X_{-}- 1^{\top}_{T} \otimes x^\star \\ U_{-} -  1^{\top}_{T}\otimes u^\star  \end{bmatrix} 
       +H W_{-} + V_{+}
 \end{equation}
 where, by Lemma~\ref{Lemma: normstate}, we set $x^\star= H^\dagger (y^\star - v^\star)$ and $v^\star = c_{\mathcal{Z}_v}$ with $y^\star$ a known linearization point. Now, let Assumption~\ref{asmp: normstate} hold. Then, again by Lemma~\ref{Lemma: normstate}, it holds that 
 \[
 X_- \in  \mathcal{M}_{x|y}=H^\dagger (Y_--\mathcal{M}_v)+g_\eta\mathcal{M}_\eta.
 \]
 To get an approximated linear model $\hat{M}_H$, we set  $X_-=H^\dagger (Y_--C_{\mathcal{M}_v})$, which is the center of matrix zonotope $\mathcal{M}_{x|y}$. Also, in \eqref{eq: Y_data}, we choose $V_{+} = C_{\mathcal{M}_v}$ and $W_{-}=C_{\mathcal{M}_w}$. 
By substituting these values in \eqref{eq: Y_data}, we obtain 
\begin{align*}
    X_{-}- 1^{\top}_{T} \otimes x^\star &= H^\dagger Y_- - H^\dagger C_{\mathcal{M}_v} - 1_T^\top \otimes H^\dagger (y^\star - v^\star)  \\
    &= H^\dagger (Y_-- 1^{\top}_{T} \otimes y^\star)
\end{align*}
because $v^\star=c_{\mathcal{Z}_v}$ and $C_{\mathcal{M}_v}=1_T^\top \otimes c_{\mathcal{Z}_v}$.
Consequently, an approximated linear model $\hat{M}_H$ can be derived using the least-squares solution of \eqref{eq: Y_data}, according to line \ref{eq: LinearMatrices} of Algorithm~\ref{alg: Reachability}.

The second step of the proof involves over-approximating the model mismatch $\Delta M_H$ and the Lagrange remainder $L_H$ for the given data. For this, we define 
\[\Xi = \begin{bmatrix} (X_-)_{.,1}&\dots&(X_-)_{.,T}\\ (U_-)_{.,1}&\dots&(U_-)_{.,T}\end{bmatrix} \] 
with $(\Xi)_{.,j}$, $j = 1,\dots,T$, denoting the $j$-th column. By substituting \eqref{eq: app_f} in \eqref{eq: in-out}, the following relation holds for any point $(\Xi)_{.,j}$ and its corresponding noise realization $(W_{-})_{.,j} \in \mathcal{Z}_w$:
  \begin{multline}
      (Y_+)_{.,j} - H(W_-)_{.,j}-(V_+)_{.,j}
      =  (\hat{M}_H+\Delta M_H)\begin{bmatrix}1\\ (\Xi)_{.,j} - \xi^\star
      \end{bmatrix}  +L_H. \label{eq: zonoeq}
 \end{multline}
 Rearranging equation \eqref{eq: zonoeq} and considering zonotopes bounding the noise realizations $\mathcal{Z}_w$, $\mathcal{Z}_v$, and the zonotope bounding the state $ \mathcal{Z}_{x|y}$ yields to bounding the model mismatch and Lagrange remainder for a single point $(\Xi)_{.,j}$ as follows:
  \begin{multline}
  \label{inc: delta+L}
  \Delta M_H\begin{bmatrix}1\\ (\Xi)_{.,j} - \xi^\star
      \end{bmatrix} + L_H \in
      (Y_+)_{.,j} - H\mathcal{Z}_w-\mathcal{Z}_v \\ -\hat{M}_H\begin{bmatrix}1\\ H^\dagger ((Y_-)_{.,j}-\mathcal{Z}_v) + g_\eta \mathcal{Z}_{\eta}-H^\dagger (y^\star - v^\star)\\ (U_-)_{.,j} - u^\star
      \end{bmatrix}.
 \end{multline}
 Now, we can over-approximate the right-hand side of \eqref{inc: delta+L} for all $(\Xi)_{.,j}$, $j = 1, \dots, T,$ corresponding to the data $D_{[t-T,t]}$ together with their state trajectory and noise realizations by $\mathcal{Z}_L$ according to lines \ref{eq: ZL_low}-\ref{eq: ZL} of Algorithm~\ref{alg: Reachability}. In particular, we have proved that for all $(\Xi)_{.,j}$, $j=1,\dots,T$, corresponding to the data $D_{[t-T,t]}$ the following holds
 \begin{equation*}
     \Delta M_H \begin{bmatrix}1\\(\Xi)_{.,j}-\xi^\star\end{bmatrix} + L_H \in   \mathcal{Z}_L.
 \end{equation*}
Next, the model mismatch and the Lagrange remainder should be over-approximated for all $\xi \in \mathcal{F}$. To this end, we assume $\mathcal{Z}_\eta$ and $\mathcal{Z}_{u}$ to be compact. Due to this assumption, all points $ (\Xi)_{.,j}$, $j=1,\dots,T,$ corresponding to $D_{[t-T,t]}$ are dense in $\mathcal{F}$, i.e., there exists some $\delta \geq 0$ such that, for any $\xi \in \mathcal{F}$, there is a $(\Xi)_{.,j}$ corresponding to $D_{[t-T,t]}$ such that $\Vert \xi-(\Xi)_{.,j}\Vert_2\le \delta$ \citep{montenbruck2016some}. The quantity $\delta$ is referred to as the covering radius of the set of $ (\Xi)_{.,j}$, $j=1,\dots,T$ corresponding to $D_{[t-T,t]}$. Given Assumption~\ref{asmp: differentiable} and a known $\delta$, we know that for every $\xi \in \mathcal{F}$, there exists a $(\Xi)_{.,j}$ corresponding to $D_{[t-T,t]}$ such that
\begin{align*}
     | f_H^{(i)}(\xi)-f_H^{(i)}((\Xi)_{.,j})| &\le \vert (H)_{i,.}\vert L_f \Vert \xi-(\Xi)_{.,j}\Vert_2 \\
     &\le  \vert (H)_{i,.}\vert L_f\delta
\end{align*}
for each dimension. Here, $\vert (H)_{i,.}\vert$ represents the element-wise absolute value of row $i$ in matrix $H$ with $i=1,\dots,n_y$. This yields 
\begin{equation*}
     \Delta M_H \begin{bmatrix}1\\\xi-\xi^\star\end{bmatrix} + L_H \in   \mathcal{Z}_L +   \mathcal{Z}_\epsilon
 \end{equation*}
and 
\begin{align}\label{eq: bound F}
f_H(\xi) \in \hat{M}_H \begin{bmatrix}1\\\xi-\xi^\star\end{bmatrix} +\mathcal{Z}_L + \mathcal{Z}_\epsilon
\end{align} 
with $\mathcal{Z}_\epsilon = \Big\langle 0_{n_y},\mathrm{diag}(\vert (H)_{1,.}\vert L_f \delta/2,\dots,\vert (H)_{n_y,.}\vert L_f\delta/2)\Big\rangle$ representing a box with length $\vert (H)_{i,.}\vert L_f\delta$
in each dimension $i$. Finally, according to equations \eqref{eq: in-out} and \eqref{eq: bound F}, the exact reachable set can be over-approximated using line \ref{eq: Y_Reach_D} of Algorithm~\ref{alg: Reachability}, i.e.,  $\hat{\mathcal{R}}^y_{t+k+1|t}  \supseteq \mathcal{R}^y_{t+k+1|t}$, which completes the proof.
\end{proof}

Note that as $T \rightarrow \infty$, i.e., $\delta$ approaches zero, we observe that $\mathcal{Z}_\epsilon \xrightarrow{} 0$. Consequently, for every $\xi \in \mathcal{F}$, $\mathcal{Z}_L$ accounts for both the model mismatch and the Lagrange remainder. For computing the Lipschitz constant and the covering radius, \citep{montenbruck2016some} and \citep{novara2013direct} have proposed methods that work well in practice.

\begin{remark}
In model-based reachability analysis of nonlinear systems, \citep{althoff2015introduction} proves that the norm of the Lagrange remainder is minimized by choosing the center of the reachable set as a linearization point. Therefore, based on this result, we simply choose centers of $\mathcal{R}^y_{t|t}$, $\mathcal{Z}_v$, and $\mathcal{Z}_{u}$, as linearization points $y^{\star}$, $v^{\star}$, and $u^{\star}$, respectively. 
\hfill $\lrcorner$
\end{remark}

\subsection{Control Phase}

In this section, we present a novel approach for solving the receding horizon optimal control problem for the nonlinear system in~\eqref{eq: sys} subject to the constraints in \eqref{eq: constraint}. At each time step $t$, we use past input-output measurements to predict future reachable regions over the time horizon $N \in \mathbb{N}$ using the learning phase described in  Section~\ref{subsec: Learning}. Then, we compute the control input trajectories such that the predicted outputs remain within the computed reachable regions while minimizing the cost over the horizon $N$.
 At time $t\in\mathbb{Z}_{\geq 0}$, given the past $T+1$ input-output data $D_{[t-T,t]}$ of the nonlinear system in \eqref{eq: sys} with constraints in \eqref{eq: constraint}, and bounds on the process and measurement noise, we define the following data-driven optimal control problem
 \begin{subequations} \label{eq: optzonopc}
\begin{alignat}{2}
&\!\min_{y,u}        &\qquad& J(y,u)=\sum_{k=0}^{N-1} \ell(y_{t+k+1|t},u_{t+k|t})
 \label{eq: Cost}\\ 
&\text{s.t.} &      &
\hat{\mathcal{R}}^y_{t+k+1|t} = \hat{M}_H (1 \times  ((\hat{\mathcal{R}}^{x}_{t+k|t} \times u_{t+k|t}) -\xi^
\star)) \nonumber\\
&&& \qquad \qquad + \mathcal{Z}_v + H \mathcal{Z}_w + \mathcal{Z}_L + \mathcal{Z}_\epsilon \label{eq: Rconst}\\
&                  &      & \hat{\mathcal{R}}^y_{t+k+1|t}  \subseteq  \mathcal{Y}_{t+k+1} \label{eq: const_y}\\
&                  &      & y_{t+k+1|t}\in \hat{\mathcal{R}}^y_{t+k+1|t} \label{eq: const_r_up}\\
&                  &      &  u_{t+k|t} \in \mathcal{U}_{t+k} \label{eq: uconst}
\end{alignat}
\end{subequations}
with the initial point $y_{t|t} = y(t)$. Here, $ u = (u_{t|t},\dots, u_{t+N-1|t}) $ and $ y = (y_{t+1|t},\dots, y_{t+N|t})$ are input and output trajectories predicted at time $t$ over the finite horizon $N$, respectively. Meanwhile, $y(t)$ is the measured output at the current time $t$. The cost function $J(y,u)$ containing a positive definite stage cost function $\ell: \mathbb{R}^{n_y} \times \mathbb{R}^{n_u}\xrightarrow{}\mathbb{R}_{\geq0} $ is minimized online. Notice that problem \eqref{eq: optzonopc} can be classified as convex provided that the stage cost function is convex. The constraint \eqref{eq: Rconst} is derived from Algorithm~\ref{alg: Reachability} where $\mathcal{Z}_u$ is  substituted with $u_{t+k|t}$. Notice that the Cartesian product of a zonotope and a vector in equation \eqref{eq: Rconst} is given by
\[
\hat{\mathcal{R}}^{x}_{t+k|t} \times u_{t+k|t}=\hat{\mathcal{R}}^{x}_{t+k|t} \times\langle u_{t+k|t},0_{n_u}\rangle.
\]
The constraint \eqref{eq: const_y} ensures that the output constraints are satisfied by predicted reachable sets \eqref{eq: Rconst}. This ultimately limits the choice of $u_{t+k|t}$ in \eqref{eq: uconst}. 
On the other hand, the constraint \eqref{eq: const_r_up} guarantees that $y_{t+k+1|t}$ remains within the allowable reachable region. We denote the optimal solution of \eqref{eq: optzonopc} at time $t$ by $u^*$ and $y^*$. We then apply the first optimal control input $u^*_{t|t}$ to the system. Similar to standard MPC, \eqref{eq: optzonopc} is solved in a receding horizon fashion as summarized in Algorithm~\ref{alg: zonopc}.

To implement the constraint \eqref{eq: const_y}, we need to verify that the predicted reachable set $\hat{\mathcal{R}}^y_{t+k+1|t}$ is a subset of the zonotope $\mathcal{Y}_{t+k+1}$. To this end, we first over-approximate $\hat{\mathcal{R}}^y_{t+k+1|t}$ by an interval vector as 
 \[
 \mathrm{int}(\hat{\mathcal{R}}^y_{t+k+1|t})=[\underline{\mathcal{I}}_{t+k+1|t},\overline{\mathcal{I}}_{t+k+1|t}]. 
 \]
Then, we verify the following constraints:
\[
 \overline{\mathcal{I}}_{t+k+1|t} \leq \overline{\mathcal{Y}}_{t+k+1}, \quad
 \underline{\mathcal{I}}_{t+k+1|t} \geq \underline{\mathcal{Y}}_{t+k+1}
\]
with $\mathrm{int}(\mathcal{Y}_{t+k+1})=[\underline{\mathcal{Y}}_{t+k+1},\overline{\mathcal{Y}}_{t+k+1}]$.
\begin{algorithm}[t]
  \caption{Nonlinear Zonotopic Predictive Control (NZPC)}
  \label{alg: zonopc}
  \textbf{Input:} Inputs of Algorithm~\ref{alg: Reachability}, time horizon $N$, input and output constraints $\mathcal{U}_k$ and $\mathcal{Y}_{k}$, and stage cost function $\ell(.)$. \\ 
  \textbf{Online Phase:}
  \begin{algorithmic}[1]
  \While{$t \in \mathbb{Z}_{\geq 0}$}
    \State\!\!\!\!\!\textbf{Learning Phase}: Compute the over-approximated reachable sets $\hat{\mathcal{R}}^y_{t+k+1|t}$ at time $t$ for $k=0,\dots,N-1$ using Algorithm~\ref{alg: Reachability}.
    \State\!\!\!\!\!\textbf{Control Phase}: Solve \eqref{eq: optzonopc} and apply the first input $u_{t|t}^*$ to the system.
    \State\!\!\!\!\!Increase the time step $t=t+1$.
    \State\!\!\!\!\!Collect the new data and update  $D_{[t-T,t]}$.
    \State\!\!\!\!\!Update the initial output set $\mathcal{R}^y_{t|t}=\langle y(t),0_{n_y}\rangle$.
  \EndWhile
  \end{algorithmic}
\end{algorithm}

In the remainder of this section, we prove the robust satisfaction of the specified control problem subject to the feasibility of \eqref{eq: optzonopc}. Notice here that, by assuming feasibility, we mean it is assumed that an admissible solution to \eqref{eq: optzonopc} exists, which satisfies the constraints. Then, the following result guarantees that such a solution can be found by Algorithm~\ref{alg: zonopc}.
\begin{theorem}\label{theorem: cont}
Consider a discrete-time nonlinear system \eqref{eq: sys} together with input and output constraints \eqref{eq: constraint}. Suppose Assumptions~\ref{asmp: normstate} and \ref{asmp: differentiable} hold, and process and measurement noise are bounded by $\mathcal{Z}_w$ and $\mathcal{Z}_v$, respectively. If problem \eqref{eq: optzonopc} is feasible at each time step, then the data-driven controller obtained from \eqref{eq: optzonopc} guarantees the robust satisfaction of the closed-loop constraints, i.e., $y(k)\in \mathcal{Y}_k$ and $u^*(k) \in \mathcal{U}_k$, for all time $k \in \mathbb{Z}_{\geq 0}$ and for all possible realizations of the bounded process and measurement noise.
\end{theorem}

\begin{proof}
Based on Theorem~\ref{theorem: Reach}, the data-driven reachable sets computed in \eqref{eq: Rconst}, over-approximates the exact reachable sets. According to the constraints \eqref{eq: const_y} and \eqref{eq: const_r_up}, the optimal control input sequence is chosen such that the output is within the intersection of the over-approximated output reachable sets and the output constraints in the presence of bounded noise, along with satisfying input constraints~\eqref{eq: uconst}. This, therefore, guarantees the robust constraint satisfaction of $\mathcal{Y}_k$ at each time step subject to the feasibility of \eqref{eq: optzonopc}.  
\end{proof}

This theorem allows optimal predictive control to be performed directly using the available noisy data, eliminating the need for an offline system identification step. The designed controller provides robust constraint guarantees against all possible bounded noise realizations. While we focused on robust constraint satisfaction, we did not tackle the issue of how to guarantee recursive feasibility and closed-loop stability in this paper. Addressing this matter is a relevant future research direction. 

\begin{remark}
    While Theorem~\ref{theorem: cont} assumes the problem~\eqref{eq: optzonopc} to be feasible and does not require additional assumptions on guaranteeing feasibility, we highlight that controllability and observability of the linearized system are necessary in this regard. The controllability condition is essential to guarantee that the system's output $y(k)$ can be regulated at any desired reference output. In contrast, observability is necessary to ensure that the state $x(k)$ does not blow up without being observed in the output $y(k)$. Furthermore, without the model knowledge and direct state measurements, verifying the controllability and observability of the nonlinear system is challenging. However, similar to \citep{berberich2022linear}, we can implicitly assume that the linearization \eqref{eq: Taylor_f} of the nonlinear system \eqref{eq: sys} is controllable and observable at every point in $\mathcal{X}\times\mathcal{Z}_u$ so that the data-driven predictive control approach proposed in this paper is feasible and well-posed. This amounts to assuming that the system is equipped with a sufficient amount of sensing and actuating equipment to ensure that the output of the nonlinear system can be steered easily and the input-output data is informative of the system's state. Formally exploring notions of controllability and observability in a data-driven setting is not considered in the current paper and is deferred for future work.
 \hfill $\lrcorner$
\end{remark}

\section{Numerical Example} \label{sec: Example}

In this section, we verify our proposed NZPC algorithm through an illustrative example. First, we demonstrate the effectiveness of the proposed data-driven reachability analysis by considering Algorithm~\ref{alg: Reachability} as a self-contained algorithm. We consider the discrete-time nonlinear system given in \citep[Section 8.3.5]{althoff2020cora}, whose model is given by:
\begingroup\makeatletter\def\f@size{9}\check@mathfonts
\def\maketag@@@#1{\hbox{\m@th\large\normalfont#1}}%
\begin{align*}
    f^{(1)}(x,u)=& \frac{(1-0.5\tau-\alpha \mathrm{exp}(\frac{\beta}{x_2(k)})\tau)x_1(k)+\tau}{1+0.5\tau} +u_1(k)\tau \\
    f^{(2)}(x,u)=& \frac{(1-1.5\tau)x_2(k)+\rho x_1(k)\mathrm{exp}(\frac{\beta}{x_2(k)})}{1+1.5\tau}\\
    & \qquad +\frac{\tau(350-6.3x_1(k)-14.4x_2(k))}{1+1.5\tau}+u_2(k)\tau
\end{align*}
\endgroup
where $\tau=0.015$, $\alpha = 7.2\cdot10^{10}$, $\beta = -8750$, and $\rho = 1.5\cdot10^{13}$. 
Note that this illustrative example differs from the original example given in \citep[Section 8.3.5]{althoff2020cora} in the sense that the process noise and the control input are applied differently, and 
we consider the system output matrix 
\[H = \begin{bmatrix} 1 & 0.001\\ -0.01 & 1 \end{bmatrix}.\]
Furthermore, we assume $\eta = 22$. The initial available input-output data contain $50$ trajectories with a length of $10$ (i.e., $T=500$). It is important to highlight that in the context of model-based reachability analysis, the function $f(x,u)$, representing the ground truth, is known. However, this knowledge is not available to the data-driven reachability analysis Algorithm and the NZPC algorithm. Additionally, in this section, we misuse the notation $\mathcal{R}^y_{t+k|t}$ to represent model-based reachable sets, which correspond to an over-approximation of the exact reachable set. For the simulations, we use CORA toolbox \citep{althoff2015introduction} in MATLAB, along with Multi-Parametric Toolbox \citep{conf:mpt} and YALMIP solver \citep{lofberg2004yalmip}. 

For data-driven reachability analysis, we assume that the initial state set is $\mathcal{X}_0 = \langle \begin{bmatrix} -2 & -20.5 \end{bmatrix}^{\top}, \mathrm{diag}(0.01,0.2) \rangle$. The process noise and measurement noise are bounded by $\mathcal{Z}_w = \langle 0_{n_x}, 2\cdot 10^{-4}1_{n_x} \rangle$ and $\mathcal{Z}_v = \langle 0_{n_y},  10^{-3}1_{n_y} \rangle$, respectively. The initial output set is $\mathcal{Y}_0 = H\mathcal{X}_0 +\mathcal{Z}_v$. The input set is $\mathcal{Z}_{u} = \langle  0_{n_u},  \mathrm{diag}(0.1,3) \rangle$, $\forall k = 0,\dots,4$. Figure~\ref{fig: MD_DD Reach} illustrates that the model-based output reachable sets $\mathcal{R}_{k+1|0}^y$ are over-approximated by data-driven output reachable sets $\hat{\mathcal{R}}_{k+1|0}^y$ computed by Algorithm~\ref{alg: Reachability} in Section~\ref{subsec: Learning}, which validates the result of Theorem~\ref{theorem: Reach}.
\begin{figure}[!htbp]
    \centering
    \includegraphics[scale=0.32]{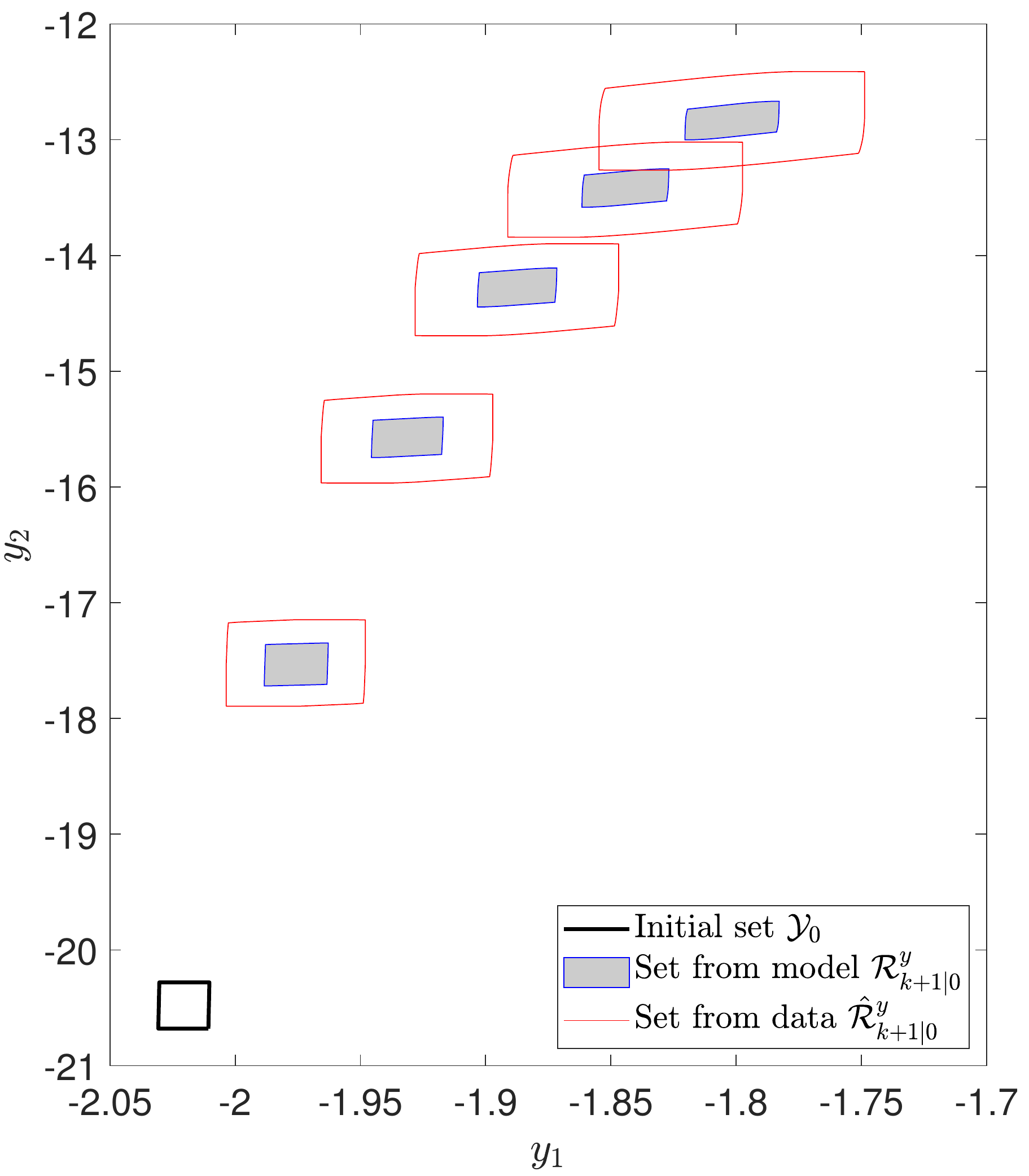}
    \caption{The output reachable sets of the numerical example computed using input-output data by applying Algorithm~\ref{alg: Reachability}.}
    \label{fig: MD_DD Reach}
\end{figure}
\begin{figure*}[!htbp]
\vspace{-0.05cm}
    \centering
    \begin{subfigure}[h]{0.32\textwidth}
     \centering
        \includegraphics[width=\textwidth]{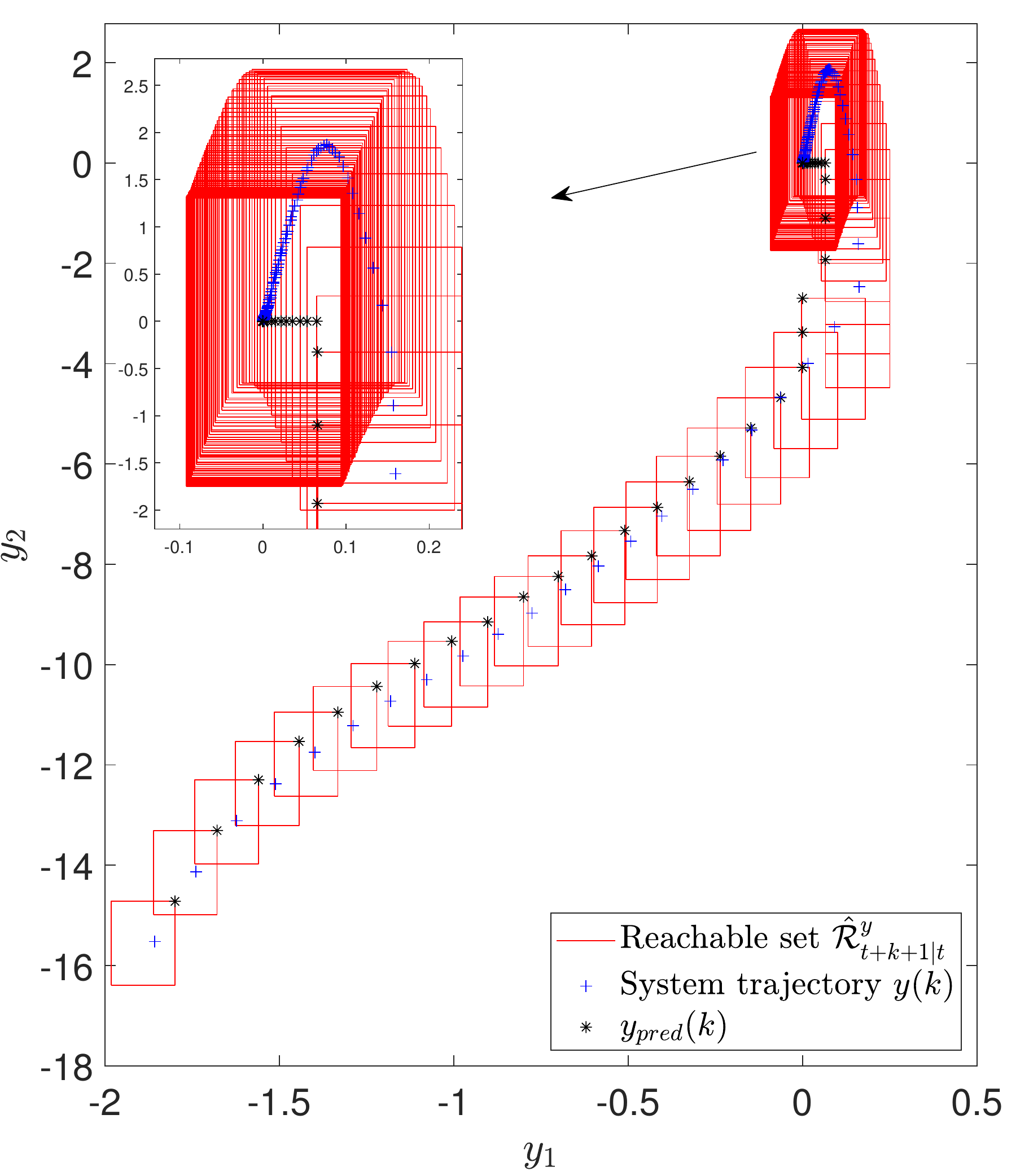}
        \caption{}
        \label{fig: ZPC_R}
    \end{subfigure}
~
    \begin{subfigure}[h]{0.31\textwidth}
     \centering
        \includegraphics[width=\textwidth]{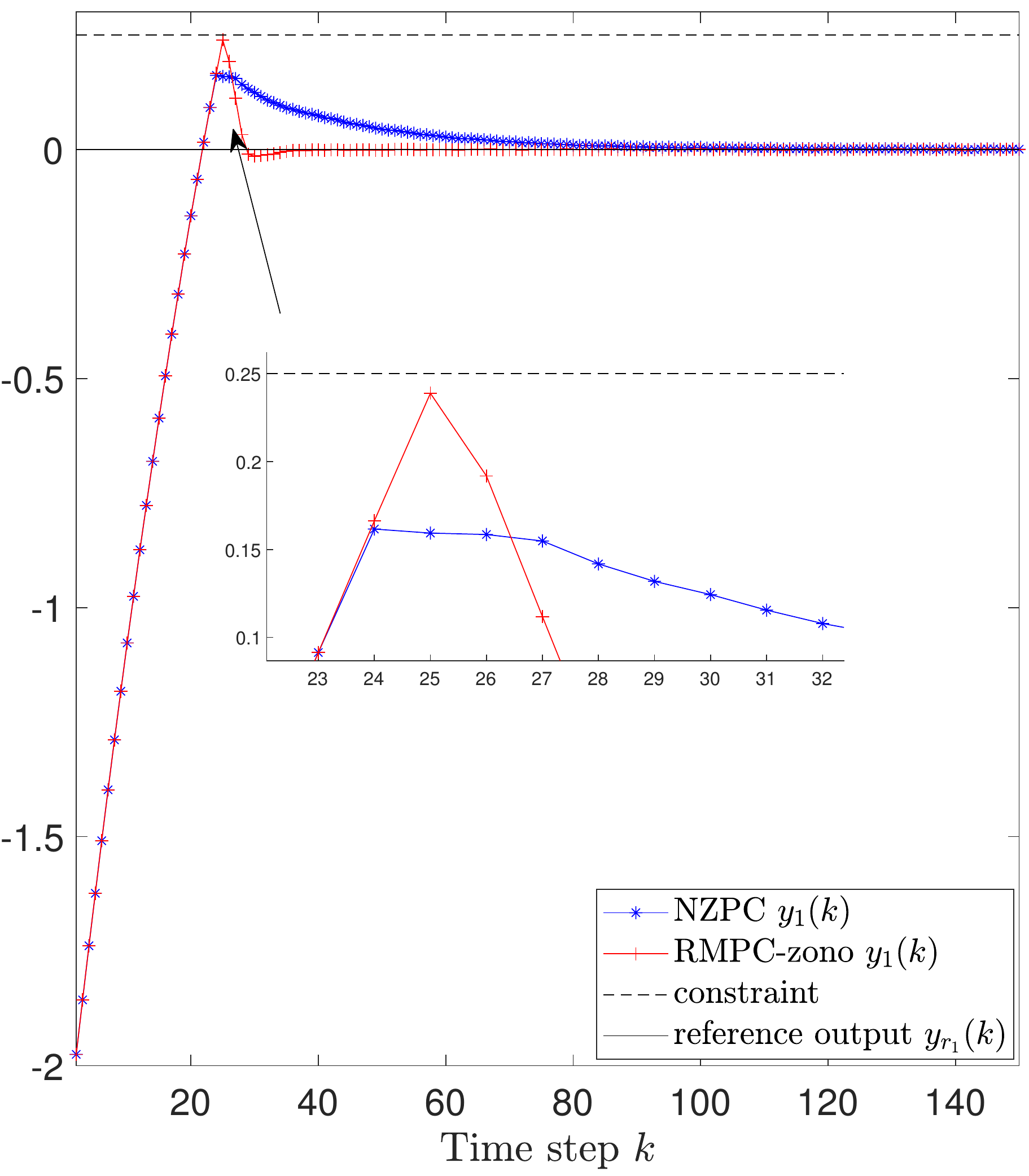}
        \caption{}
        \label{fig: Y_1}
    \end{subfigure}
~
    \begin{subfigure}[h]{0.32\textwidth}
     \centering
        \includegraphics[width=\textwidth]{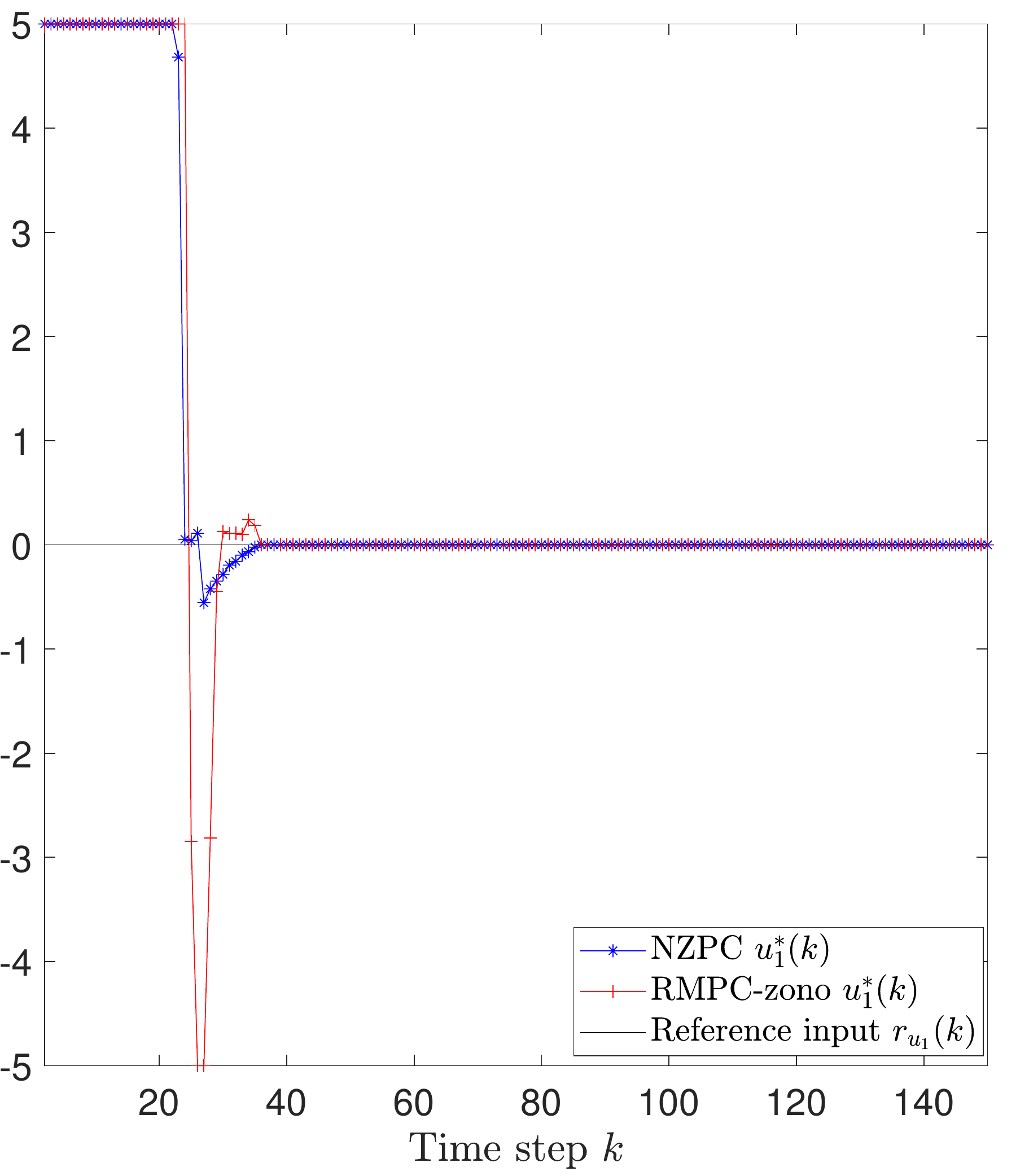}
        \caption{}
        \label{fig: U_1}
    \end{subfigure}
\caption{The reachable sets over 150 time steps of the closed-loop system employing the NZPC approach are presented in (a). One dimension of the closed-loop output trajectory $y_1(k)$ using both NZPC and RMPC-zono approaches is depicted in (b). One dimension of the control input $u_1(k)$ computed utilizing NZPC and RMPC-zono techniques is illustrated in (c).}
    \label{fig:contreach}
     \vspace{-6mm}
\end{figure*}

Next, we apply the proposed NZPC approach (Algorithm~\ref{alg: zonopc}) to the same system and demonstrate its applicability considering the following cost function: 
\[
J(y,u)=\sum_{k=0}^{N-1}\norm{y_{t+k+1|t} -y_r}_Q^2
 +\norm{u_{t+k|t} - u_r}_R^2
 \]
where $y_r = 0_{n_y}$, $u_r = \begin{bmatrix} 0 & 0.007 \end{bmatrix}^{\top}$, $Q = 5I_{n_y}$, and $R = 0.02I_{n_u}$.
We set the following parameters:
\begin{align*} 
\mathcal{X}_0 &= \langle \begin{bmatrix} -2 & -20.5 \end{bmatrix}^{\top}, \mathrm{diag}(0.01,1) \rangle, \\
\mathcal{Z}_w &= \langle 0_{n_x},
\begin{bmatrix} 0.0002 & 0.02\end{bmatrix}^{\top}\rangle, \\ \mathcal{Z}_v &= \langle 0_{n_y},\begin{bmatrix} 0.001 & 0.01 \end{bmatrix}^{\top} \rangle,
\end{align*}
and $N =3$. The constraints are $\mathcal{U} = \langle  u_r,  \mathrm{diag}(5,3) \rangle$, $\mathcal{Y}_l = \begin{bmatrix} -3 &-22 \end{bmatrix}^{\top}$, $\mathcal{Y}_u = \begin{bmatrix} 0.25 & 2.7 \end{bmatrix}^{\top}$, and $ \mathcal{Z}_\epsilon = \langle 0_{2},\mathrm{diag}(0.08,0.71)\rangle$. The data-driven reachable sets $\hat{\mathcal{R}}_{t+k+1|t}^y$, system trajectory $y(k)$, and predicted output $y_{pred}(k)$ of the closed-loop system are plotted in Figure~\ref{fig: ZPC_R} over 150 time steps. Figure~\ref{fig: ZPC_R} illustrates that the output trajectory of the closed-loop system and the predicted outputs by NZPC are inside the reachable sets.
For comparison, we also apply NZPC assuming the availability of an accurate model of the system, which we refer to as the RMPC-zono approach. The input-output trajectories of one dimension of the closed-loop system under both the NZPC and RMPC-zono approaches are depicted in Figures \labelcref{fig: Y_1,fig: U_1}. The results indicate that in NZPC, the convergence of the closed-loop outputs towards the output reference is slower compared to RMPC-zono. This can be attributed to the less accurate prediction in NZPC. However, it is noteworthy that both control approaches satisfy all constraints.

We would like to note that while we did not explicitly discuss the computation of the reference control input in this paper, it is worth mentioning that similar techniques, as described in \citep{mishra2021deep}, were employed for its computation. Moreover, there are alternative methods proposed in the literature. For instance, authors in \citep{berberich2022linear} proposed a different stage cost function by introducing an artificial equilibrium that is optimized online, and it does not require prior knowledge of whether a given input-output setpoint is a feasible equilibrium. 

We presented the effectiveness of the NZPC approach th-rough a simple example for illustration purposes. However, this approach can be applied to higher-dimensional systems, as well as those in which the number of outputs is less than the number of states. In particular, our future work will exploit the presented results to design a data-driven predictive scheme for complex systems like smart buildings to optimize energy use for space heating while maintaining the thermal comfort levels of the occupants \citep{molinari2023using}.

\section{Conclusion} \label{sec: con}

We presented a zonotopic data-driven predictive control approach to robustly control unknown nonlinear systems using only input-output data, where the output data could be noisy. The proposed algorithm uses available data to over-approximate reachable sets over a finite horizon through the learning phase (Algorithm~\ref{alg: Reachability}). Then, in the control phase, we utilized the over-approximated reachable sets as an implicit data-driven system representation in the receding horizon optimal control problem. The proposed control algorithm updates the input-output data as the closed-loop system evolves (Algorithm~\ref{alg: zonopc}). We showed that the optimal control inputs computed by our proposed algorithm provide robust system constraint satisfaction. 

As compared to the existing literature, the distinctive features of our method are twofold. First, we provide robust safety guarantees in the presence of bounded process and measurement noise. Second, these guarantees do not require explicit knowledge of the nonlinear system model.  

The work presented in this paper points towards several open problems to be addressed in future works. For instance, future work involves providing necessary and sufficient conditions for the recursive feasibility of the NZPC algorithm. Related to this, characterizing data-driven observability and controllability notions for systems with unknown models is a challenging problem that needs to be addressed. In addition, extending the current framework to include unknown nonlinear systems with unknown output map is an interesting problem that will also be pursued. Not knowing the output map makes the problem quite challenging because, in this case, the dimension of the state variable  may not be accurately inferred. We believe that the method proposed in this paper will serve as the foundation for the aforementioned prospects.

\section*{Acknowledgement}
This research was funded by the Swedish Energy Agency and IQ Samh{\"a}llsbyggnad, under the E2B2 programme, grant agreement 2018-016237, project number 47859-1 (Cost- and Energy-Efficient Control Systems for Buildings), by the Swedish Foundation for Strategic Research (SSF) under grant agreement No. RIT17-0046 (Project CLAS—Cybers{\"a}kra l{\"a}rande reglersystem), by Digital Futures (Project HiSS—Humanizing the Sustainable Smart City), grant agreement VF-2020-0260, and by European Union's Horizon Research and Innovation Programme under Marie Sk\l{}odowska-Curie grant agreement No. 101062523.

\bibliography{ref}

\end{document}